\documentclass[12pt]{article}

\usepackage[english]{babel}

\usepackage[onehalfspacing]{setspace}

\usepackage{amssymb}
\usepackage{dsfont}
\usepackage{amsmath}
\usepackage{amsthm}
\usepackage{siunitx}
\PassOptionsToPackage{hyphens}{url}\usepackage{hyperref}
\usepackage{cleveref}
\usepackage[utf8]{inputenc}
\usepackage[right]{lineno}
\usepackage{csquotes}
\usepackage{booktabs}
\usepackage{longtable}
\usepackage{adjustbox}
\usepackage{array}
\usepackage{url}
\usepackage{titlesec}
\usepackage{graphicx}
\usepackage{float}
\usepackage{wrapfig}
\usepackage{subcaption}
\usepackage{tabularx}
\usepackage{authblk}
\usepackage{xcolor} 

\titleformat{\subsection}
  {\mdseries\large} 
  {\thesubsection}{1em}{} 

\newtheorem{lem}{Lemma}
\newtheorem{definition}{Definition}

\newtheorem{corollary}{Corollary}

\newtheorem{theorem}{Theorem}
\newtheorem{example}{Example}

\usepackage[english]{babel}
\usepackage[style=authoryear,backend=biber,natbib=true,maxcitenames=2,uniquelist=false,maxbibnames=5]{biblatex}
\DeclareNameAlias{sortname}{family-given}
\DeclareNameAlias{default}{family-given}

\renewbibmacro{in:}{}
\DeclareFieldFormat[article]{title}{\mkbibquote{#1}\addcomma}
\DeclareFieldFormat[book]{title}{\mkbibemph{#1}\addcomma}
\DeclareFieldFormat[bookinbook]{title}{\mkbibemph{#1}\addcomma}
\DeclareFieldFormat[inbook]{title}{\mkbibquote{#1}\addcomma}
\DeclareFieldFormat[incollection]{title}{\mkbibquote{#1}\addcomma}
\DeclareFieldFormat[inproceedings]{title}{\mkbibquote{#1}\addcomma}
\DeclareFieldFormat[manual]{title}{\mkbibemph{#1}\addcomma}
\DeclareFieldFormat[misc]{title}{\mkbibemph{#1}\addcomma}
\DeclareFieldFormat[thesis]{title}{\mkbibemph{#1}\addcomma}
\DeclareFieldFormat[unpublished]{title}{\mkbibquote{#1}\addcomma}
\DeclareFieldFormat[patent]{title}{\mkbibemph{#1}\addcomma}
\DeclareFieldFormat[report]{title}{\mkbibemph{#1}\addcomma}
\DeclareFieldFormat[online]{title}{\mkbibquote{#1}\addcomma}
\DeclareFieldFormat[software]{title}{\mkbibemph{#1}\addcomma}
\DeclareFieldFormat[booklet]{title}{\mkbibemph{#1}\addcomma}
\DeclareFieldFormat[periodical]{title}{\mkbibemph{#1}\addcomma}
\DeclareFieldFormat[standard]{title}{\mkbibemph{#1}\addcomma}

\DeclareFieldFormat[article]{journaltitle}{\iffieldundef{shortjournal}{\mkbibemph{#1}\addcomma}{\mkbibemph{\printfield{shortjournal}}\addcomma}}
\DeclareFieldFormat{volume}{\bibstring{volume}~#1}
\DeclareFieldFormat{number}{\bibstring{number}~#1}

\DefineBibliographyStrings{english}{
  volume = {Vol.},
  number = {No.}
}

\renewbibmacro*{volume+number+eid}{%
  \printfield{volume}%
  \setunit*{\addspace}%
  \printfield{number}%
  \setunit{\addcomma\space}%
  \printfield{eid}}

\renewbibmacro*{journal+issuetitle}{%
  \usebibmacro{journal}%
  \setunit*{\addcomma\space}%
  \usebibmacro{volume+number+eid}%
  \setunit{\addcomma\space}%
  \usebibmacro{issue+date}}

\renewbibmacro*{publisher+location+date}{%
  \printlist{publisher}%
  \iflistundef{location}
    {\setunit*{\addcomma\space}}
    {\setunit*{\addcolon\space}}%
  \printlist{location}%
  \setunit*{\addcomma\space}%
  \usebibmacro{date}}

\DeclareCiteCommand{\cite}[\mkbibparens]
  {\usebibmacro{prenote}}
  {\usebibmacro{citeindex}%
   \usebibmacro{cite}}
  {\multicitedelim}
  {\usebibmacro{postnote}}

\renewbibmacro*{cite:labelyear+extrayear}{%
  \iffieldundef{labelyear}
    {}
    {\printtext[bibhyperref]{%
       \printfield{labelyear}%
       \printfield{extrayear}}}}

\renewbibmacro*{cite:labeldate+extradate}{%
  \iffieldundef{labelyear}
    {}
    {\printtext[bibhyperref]{%
       \printfield{labelyear}%
       \printfield{extradate}}}}

\AtEveryBibitem{
  \clearfield{month}
  \clearfield{day}
  \ifentrytype{book}{
    \clearlist{location}
  }{}
}

\DefineBibliographyStrings{english}{
  andothers = {\text{et al.},}
}

\DeclareFieldFormat[article]{volume}{\bibstring{jourvol}\addnbspace #1}
\DeclareFieldFormat[article]{number}{\bibstring{number}\addnbspace #1}
\DeclareFieldFormat[article]{volume}{Vol. #1}
\DeclareFieldFormat[article]{number}{No. #1}

\DeclareFieldFormat{urldate}{\mkbibparens{accessed \addspace#1}}

\DeclareFieldFormat{urldate}{%
  \mkbibparens{accessed\space%
    \thefield{urlday}\addspace%
    \mkbibmonth{\thefield{urlmonth}}\addspace%
    \thefield{urlyear}}}

\crefformat{figure}{#2Figure~#1#3}
\Crefformat{figure}{#2Figure~#1#3}
\crefformat{table}{#2Table~#1#3}
\Crefformat{table}{#2Table~#1#3}
\crefformat{section}{#2Section~#1#3}
\Crefformat{section}{#2Section~#1#3}

\author[1]{Jonathan Klinge}
\author[2]{Maren Diane Schmeck}
\affil[1,2]{Center for Mathematical Economics, Bielefeld University}
\title{Asymptotics of Ruin Probabilities in a Subordinated Cramér-Lundberg Model\footnote{Financial support by the German Research Foundation (DFG) [RTG 2865/1 – 492988838] is gratefully acknowledged. \\
The authors thank Alfred Müller, Matthias Scherer and Hanspeter Schmidli for helpful comments and fruitful discussions.}}

\begin{document}
\nocite{}
\maketitle

\begin{abstract}
\noindent
We study a dynamic model of a non-life insurance portfolio. The foundation of the model is a compound Poisson process that represents the claims side of the insurer. To introduce clusters of claims appearing, e.g. with catastrophic events, this process is time-changed by a Lévy subordinator. The subordinator is chosen so that it evolves, on average, at the same speed as calendar time, creating a trade-off between intensity and severity.
We show that such a transformation always has a negative impact on the probability of ruin. Despite the expected total claim amount remaining invariant, it turns out that the probability of ruin as a function of the initial capital falls arbitrarily slowly depending on the choice of the subordinator. 
\end{abstract}

\newpage
\maketitle
\noindent
\section{Introduction}
Insurers today face a wide range of risks that need to be modeled appropriately. On the one hand, there are individual risks and associated individual losses, such as vehicle damage covered by comprehensive motor insurance in the case of an accident. On the other hand, there are natural hazards such as hail, storms, or earthquakes. In these cases, the insurer records cumulative losses arising from many individual claims that occur simultaneously (e.g., hail damaging numerous vehicles at once). Due to climate change, natural catastrophes and their financial consequences have gained increasing relevance, making the modeling of cumulative risks more important than ever.
\\
\\
Cumulative losses caused by natural disasters have a substantial impact on the insurer’s probability of ruin, as they occur at a single point in time rather than being spread over the duration of the insurance contract, as is typical for independent individual losses. The magnitude of this risk depends not only on climatic conditions, but also on the insurer’s exposure concentration. Portfolios with high exposure to natural hazards and/or a strong geographical concentration are exposed to a higher cumulative risk.
\\
\\
In this paper, we construct a model that captures the effect of cumulative losses on the probability of ruin. We consider a compound Poisson process describing the insurer’s loss side and introduce a Lévy subordinator that acts as a stochastic time change of this process. This time change allows the base process to “jump” over certain periods, so that what are individual losses in ordinary time may cluster into a large aggregate loss in the time-changed process. The subordinated loss process therefore exhibits random clustering of individual claims. The model extends \citet{selchscherer}, though we focus on a different application.
The theory of subordination of Lévy processes has been introduced and studied in \citet{bochner1949}, \citet{bochner1955} and developed further. Rich sources on the theory of subordination are, for example, \citet{sato1999}, \citet{conttankov}, and \citet{applebaum2009}.
Specifically, subordination of compound Poisson processes is examined in \citet{Di_Crescenzo_2015}, and \citet{sengar2020}.
\\
\\
The economic interpretation of the stochastic time change in our model is straightforward: a large jump of the subordinator corresponds to the occurrence of a major catastrophic event.
To ensure comparability between the subordinated process and the original compound Poisson process, we require the subordinator to run on average at the same speed as calendar time. Consequently, the expected value of the compound Poisson process remains unchanged under subordination. Subordination therefore acts purely as a random distortion of the claim arrival times. Mathematically, time normalization induces a trade-off between jump intensity and jump size such that the expected value remains invariant. In addition, we demonstrate that subordination invariably enlarges the jump size distribution of the compound Poisson process in the sense of stochastic dominance and that, through the clustering of individual jumps, subordination can transform light-tailed jump distributions into heavy-tailed ones. A detailed discussion of heavy-tailed distributions and their subclasses can be found in sources such as \citet{bingham1989}, \citet{fosskorshunov}, and \citet{ekm2013}.
\\
\\
The paper is organized as follows.
First we show that the subordinated compound Poisson process is again a compound Poisson process with modified intensity and jump distribution. We derive conditions under which the modified jump distribution has light tails, allowing the application of classical Cramér-Lundberg arguments. Nevertheless, we demonstrate that the modified jump distribution always dominates the original distribution in the sense of first-order stochastic dominance. In the light-tailed case we prove, via the Adjustment Coefficient, that subordination asymptotically increases the ruin probability.
\\
\\
We then turn to the heavy-tailed case and identify conditions under which the modified jump distribution becomes heavy-tailed. Our analysis focuses on subexponential distributions, in particular those with regularly varying tails, and we derive conditions for the subordinated process to inherit these tail properties.
\\
\\
As mentioned above, several sources address the subordination of Lévy processes, particularly within the actuarial literature. A more recent development in this area is, for example, \citet{ZHANG2020166}, which studies the valuation of joint-life products and models mortality dependence via stochastically dependent subordinators.
\\
Important sources on ruin theory and the Cramér-Lundberg model include, for example, \citet{rolski2009stochastic}, \citet{asmussenalbrecher}, and \citet{schmidli2017}.
\\
Although ruin theory is a very old field, dating back to \citet{lundberg1903approximerad}, it remains an active area of contemporary research. Recent contributions include \citet{ALBRECHER20251}, \citet{LOCAS2025189}, \citet{CHEUNG202384}, \citet{LI202444}, \citet{lindskog2025eliciting} and \citet{zhu2025ruin}. In \citet{ALBRECHER20251}, an optimal dividend problem is discussed under the condition that the probability of ruin remains sufficiently small. \citet{LOCAS2025189}, \citet{CHEUNG202384}, and \citet{LI202444} discuss a weakened concept of ruin, known as Parisian ruin, where ruin only occurs if the surplus process is negative for a sufficiently long period. This concept of ruin is analysed in both the Cramér-Lundberg model and the Sparre Andersen model. \citet{lindskog2025eliciting} analyzes a claims reservation problem with backlog, while \citet{zhu2025ruin} discusses ruin theory under dependency.
On recent research on the Cramér-Lundberg model \citet{AVERHOFF2025103128} should be mentioned. Here, a Cramér-Lundberg model with unknown parameters is considered. \citet{BRAUNSTEINS20231} analyzes a Cramér-Lundberg model in which the number of insureds is not constant, but each customer has a certain entry rate and random length to stay. Papers dealing specifically with Nat-Cat and Cat Bonds include \citet{FAIAS202046}, \citet{COLANERI2021498}, \citet{BARADEL202416}, and \citet{LIU20251}.
\newpage
\section{Subordinated Compound Poisson Process}
\subsection{Model}
We consider the following model. Let $(\Omega,\mathcal{F},(\mathcal{F}_t)_{t\in[0,\infty)},\mathbb{P})$ be a filtered probability space. Let $(N_t)_{t\geq0}$ be a standard Poisson process with intensity $\lambda$. The claim sizes $X_i$ for $i\in\mathbb{N}$, are independent and identically distributed positive random variables.
\\
Consider the compound Poisson process
\begin{equation*}
C_t=\sum_{i=1}^{N_t}{X_i},
\end{equation*}
with Lévy measure $\nu_C$ and distribution $\mu^t:=\mathcal{L}(C_t|\mathbb{P})$ at time $t\geq0$. Additionally we define an independent univariate Lévy subordinator $(\Lambda_t)_t$ on $(\Omega,\mathcal{F},(\mathcal{F}_t)_{t\in[0,\infty)},\mathbb{P})$ with Lévy triplet $(0,\nu_\Lambda,b_\Lambda)$.
Using these ingredients, we define the main object of our analysis, the subordinated compound Poisson process
\begin{equation}\label{defsubordinatedprocess}
    Y_t:=C_{\Lambda_t}=\sum_{i=1}^{N_{\Lambda_t}}X_i,\quad t\geq 0.
\end{equation}
\\
The following graph (Figure \ref{Trajectory subordinated}) provides an economic interpretation of the subordinated process. Due to the stochastic time change and in particular because the subordinator may exhibit jumps, several individual claims get clustered into a single cumulative jump of 
$(Y_t)_{t\geq0}$. Such a jump can be interpreted as a NatCat-type claim, representing a random aggregate of many underlying losses. The dotted lines describe the time jumps caused by the subordinator.
\\
\\
\begin{figure}
\begin{subfigure}[c]{0.32\textwidth}
\includegraphics[width=1\textwidth]{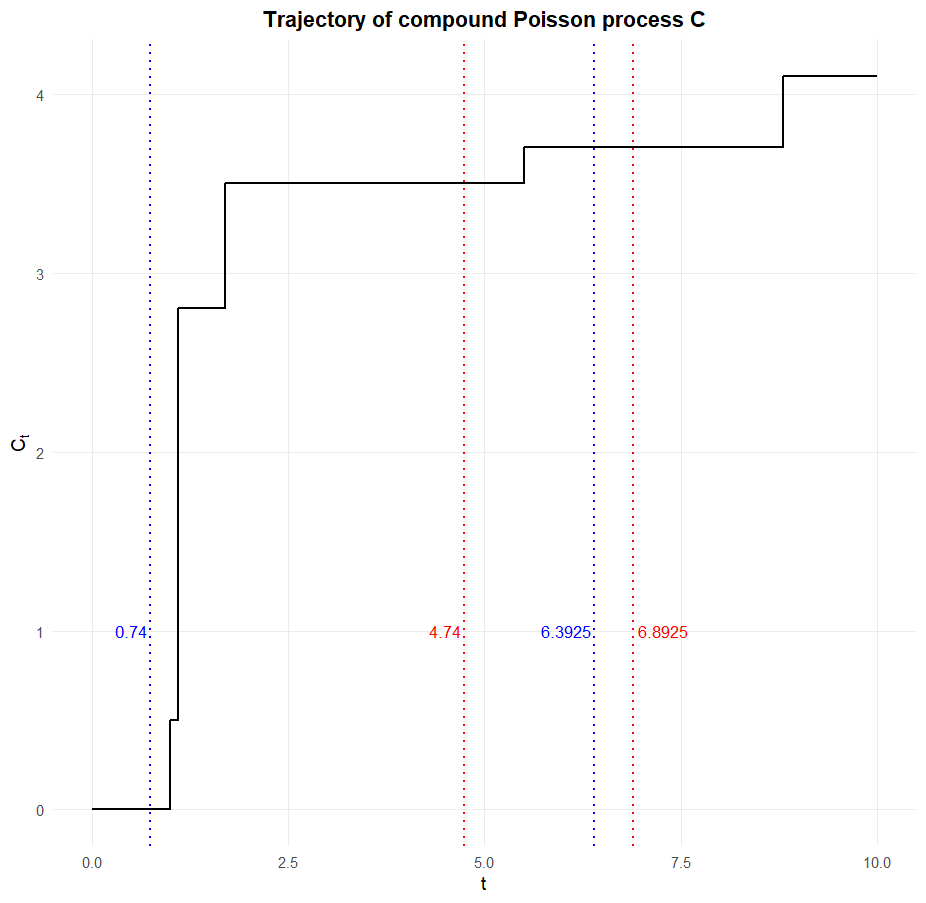}
\end{subfigure}
\begin{subfigure}[c]{0.32\textwidth}
\includegraphics[width=1\textwidth]{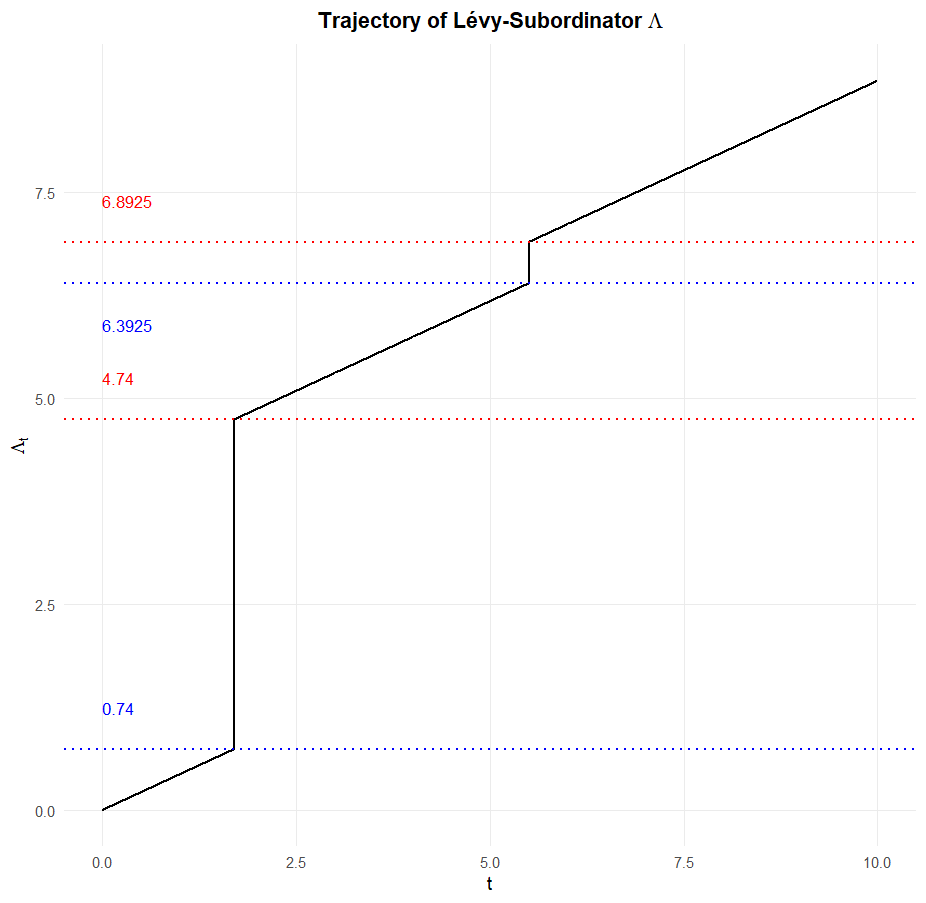}
\end{subfigure}
\begin{subfigure}[c]{0.32\textwidth}
\includegraphics[width=1\textwidth]{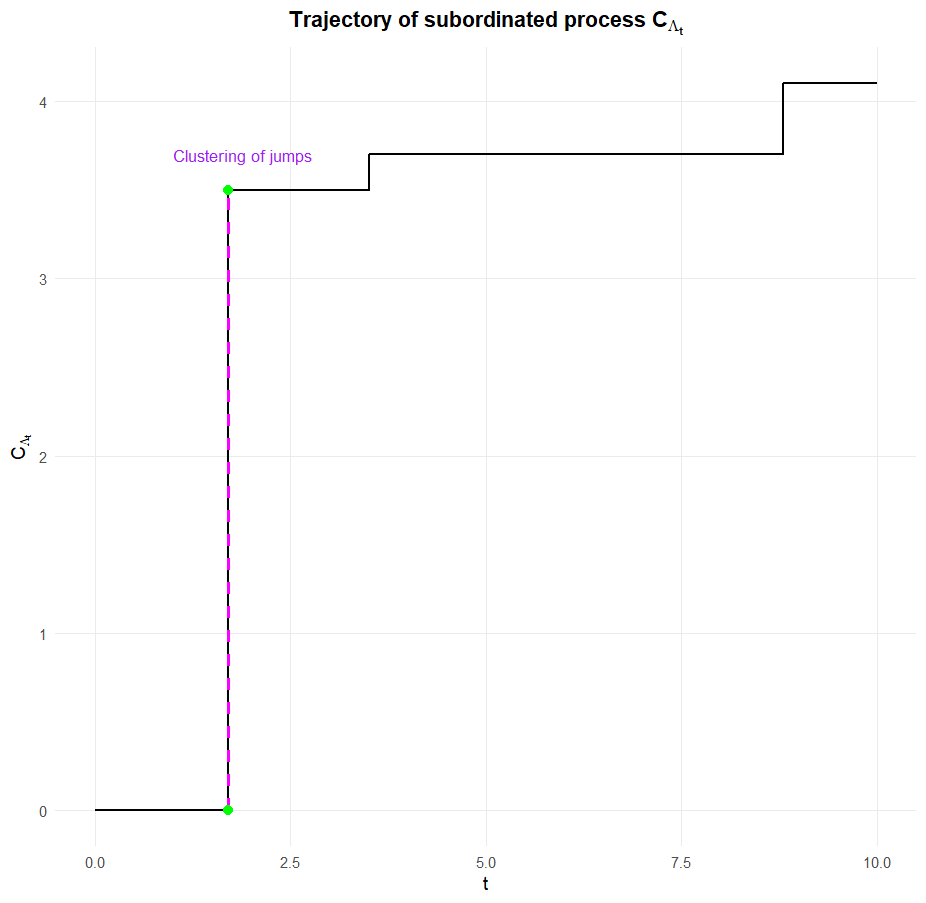}
\end{subfigure}
\caption{Trajectory of a subordinated compound Poisson process, subordinated by drifted
compound Poisson process}
\label{Trajectory subordinated}
\end{figure}
\\
\newpage
\noindent
One should notice that $(Y_t)_{t\geq0}$ is not necessarily adapted to the Filtration $(\mathcal{F}_t)_{t\geq 0}$. We therefore define the filtration $(\mathcal{G}_t)_{t\geq 0}:=(\mathcal{F}_{\Lambda_t})_{t\geq 0}$ to ensure that $Y_t$ is $\mathcal{G}_t$-measurable for every $t\in[0,\infty)$.
\\
\\
In this article, we restrict ourselves to Lévy subordinators that satisfy the time normalization condition $\mathbb{E}[\Lambda_t]=t$ or equivalently $\mathbb{E}[\Lambda_1]=1$. The motivation for this condition is that it ensures that 
\begin{equation*}
\mathbb{E}[Y_t]=\mathbb{E}[C_t],\quad t\geq 0,
\end{equation*}
which follows immediately by conditioning. Therefore, the subordinated model remains directly comparable to the baseline model.
\\
\\
As a first step, we recall the Laplace transform of a Lévy subordinator $\Lambda$. For any Lévy subordinator, it holds that for any $t\geq 0$,
\begin{equation}\label{laplace}
    \mathcal{L}[\Lambda_t](u):=\mathbb{E}[e^{-u\Lambda_t}]=e^{-t\psi_{\Lambda}(u)}, 
\end{equation}
where the Laplace exponent is given by
\begin{equation}\label{laplaceexpo}
    \psi_\Lambda(u)=b_{\Lambda}u+\int_{(0,\infty)}(1-e^{-ux})\nu_\Lambda(dx),\quad u\geq 0,
\end{equation}
for drift $b_\Lambda$ and Lévy measure $\nu_\Lambda$; see for instance \citet{bertoin1998}.
\\
From \citet{sato1999} it is known that $Y$, as defined above, is again a Lévy process. In fact we show in the following, that $Y$ is even a compound Poisson process, with explicitly known intensity and jump size distribution. 
As we were unable to identify a formulation of this result in the literature that matches the compactness required for our purposes, we present it here together with a detailed proof. This result is of fundamental importance for the arguments developed in this paper.
\begin{theorem}\label{SubCPP}
    The process $(Y_t)_{t\geq 0}:=(C_{\Lambda_t})_{t\geq 0}$, defined in (\ref{defsubordinatedprocess}), has a compound Poisson representation with intensity $\psi_{\Lambda}(\lambda)$ for the Laplace exponent $\psi_{\Lambda}$ of $\Lambda$ and jump size distribution
    \begin{equation}\label{Zdistr}
        \frac{\nu_Y(dx)}{\psi_{\Lambda}(\lambda)},
    \end{equation}
    where $\nu_Y$ is the Lévy measure of $Y$, which is given by
    \begin{equation}\label{YLevymeasure}
        \nu_Y(dx)=b_\Lambda\nu_C(dx)+\int_{(0,\infty)}\mu^t(dx)\nu_{\Lambda}(dt).
    \end{equation}
\end{theorem}
\begin{proof}
    By Lemma 2.2 in \citet{selchscherer}, we know that $Y$ is a Lévy subordinator and has the characteristics
    \begin{enumerate}
        \item $b_Y=b_\Lambda b_C$
        \item $\nu_Y(B)=b_\Lambda\nu_C(B)+\int_{(0,\infty)}{\mu^t(B)\nu_{\Lambda}(dt)},\quad \forall B\in\mathcal{B}(\mathbb{R}\setminus\{0\})$. 
    \end{enumerate}
Here $\mu^t:=\mathcal{L}(C_t\vert\mathbb{P})$, for every $t\geq 0$. The quantities $b_{\Lambda}$, $b_Y$ and $b_C$ denote the drift of $\Lambda$, $Y$ and $C$, respectively, while $\nu_{\Lambda}$, $\nu_Y$ and $\nu_C$ denote their Lévy measures.
Since $C$ is a compound Poisson process, we know that $b_C=0$ and therefore $b_Y=0$. Hence, the characteristic exponent of the subordinator $Y$ is given by
\begin{equation*}
\Psi_Y(u)=\int_0^{\infty}\left(e^{iux} -1\right)\nu_Y(dx).
\end{equation*}
To prove that $Y$ is a compound Poisson process, we have to show that the measure $\nu_Y$ is finite. We get
\begin{equation*}
    \int_{(0,\infty)}{\nu_Y(dx)}=\nu_{Y}((0,\infty))={b_\Lambda\nu_{C}((0,\infty))}+
    \int_{(0,\infty)}\mu^t((0,\infty))\nu_\Lambda(dt).
\end{equation*}
The first part of the addition is straightforward, since $X_1$ is positive and $\nu_C$ is the Lévy measure of the compound Poisson process $C$, which is given by $\nu_C(S)=\lambda\mathbb{P}[X_1\in S]$ for every $S\in\mathcal{B}(\mathbb{R}\setminus\{0\})$. Combining this with the second part of the sum, we get
\begin{equation*}
    \int_{(0,\infty)}{\nu_Y(dx)}=b_\Lambda\lambda+\int_{0}^\infty{\mathbb{P}[C_s> 0]}\nu_\Lambda(ds)=b_\Lambda\lambda+\int_{0}^\infty{(1-e^{-\lambda s})\nu_\Lambda(ds)}<\infty.
\end{equation*}
The finiteness follows from the inequality $1-e^{-\lambda s}\leq \min\{\lambda s,1\}$ for all $s\geq 0$, which is integrable with respect to $\nu_\Lambda$, since $\Lambda$ is a Lévy subordinator. Moreover from the calculation we see that
\begin{equation*}
  \int_{(0,\infty)}{\nu_Y(dx)}=\psi_\Lambda(\lambda),  
\end{equation*}
where $\psi_\Lambda$ is the Laplace exponent of $\Lambda$ as defined in (\ref{laplaceexpo}).
\\
Therefore, by normalizing $\nu_Y$, we obtain that $Y$ has characteristic exponent
\begin{equation*}
    \Psi_Y(u)=\psi_\Lambda(\lambda)\int_{(0,\infty)}(e^{iux} -1)\frac{\nu_Y}{\psi_\Lambda(\lambda)}(dx),
\end{equation*}
which proves that $Y$ is a compound Poisson process with the characteristics specified in the theorem.
\end{proof}
\noindent
Theorem \ref{SubCPP} shows that subordination of a compound Poisson process leads again to a compound Poisson process. In addition, the theorem provides the intensity and the severity distribution of the compound Poisson process $Y$, which can be read directly from its characteristic exponent.
\\
For further analysis, we write
\begin{equation}\label{equaldistr}
    Y_t:=\sum_{i=1}^{N_{\Lambda_t}}X_i\overset{\mathcal{L}}{=}\sum_{j=1}^{\tilde{N}_t}Z_j,\quad\forall t>0,
\end{equation}
where $\tilde{N}$ is a standard Poisson process with intensity $\psi_{\Lambda}(\lambda)$, independent of $Z_j$ for every $j\in\mathbb{N}$. The random variables $(Z_j)_{j\in\mathbb{N}}$ are independent and identically distributed, such that the distribution of $Z_1$ is given by \eqref{Zdistr}. Moreover $\psi_{\Lambda}$ is defined as in \eqref{laplaceexpo}. 
For notational simplicity, we define $X\sim X_1$ and write $X$ instead of $X_1$ for a independent copy of $X_1$. Analogously, we define $Z \sim Z_1$ and write $Z$ instead of $Z_1$ for an independent copy $Z$ of $Z_1$.
The distribution of $Z$, as defined in \eqref{Zdistr} and \eqref{YLevymeasure} can be interpreted as a mixture distribution of individual jumps of $X$ and clustered jumps. The first summand in \eqref{YLevymeasure} therefore represents the individual jumps, which follow the distribution of $X$, while the latter part represents the jumps clustered by subordination, since clustering can only occur in the case of a jump in $\Lambda$.
\\
Due to the non-explicit structure of the distribution of $Z$ and the compound Poisson distribution $\mu^t$, there are very few examples in which the distribution of $Z$ can be specified explicitly.
A very easy example is the following. 
\begin{example}
Define
\begin{equation*}
    \Lambda_t:=\frac{1}{2}t+\frac{1}{2}K_t,
\end{equation*}
where $(K_t)_t$ is a standard Poisson process with intensity $1$. It is evident that $\Lambda$ is time-normalized. Denoting the distribution of $Z$ by $P_Z$, it is easy to calculate, that for any $B\in\mathcal{B}(\mathbb{R}\setminus\{0\})$,
\begin{equation}\label{examplestuff}
    \psi_{\Lambda}(\lambda)P_Z(B)=\frac{1}{2}\lambda\mathbb{P}[X\in B]+\mu^{\frac{1}{2}}(B),
\end{equation}
as in this example, $\nu_{\Lambda}(dt)=\mathds{1}_{\frac{1}{2}}(dt)$, which simplifies the integral in \eqref{YLevymeasure} drastically. In the above, $\psi_{\Lambda}(\lambda)=1+\frac{\lambda}{2}-\exp(-\frac{\lambda}{2})$.
\eqref{examplestuff} illustrates the structure of the distribution of $Z$. If the subordinator does not jump, the jump heights are distributed as $X$, and the intensity is given by $\frac{1}{2}\lambda$. If, on the other hand, the subordinator jumps, its jump height in our example is $1/2$. This means that time jumps by a distance of $1/2$, and therefore the clustered jump height has distribution $\mu^{\frac{1}{2}}$, which happens on average once per unit of time. Hence, the distribution of $Z$ can be interpreted as a mixed distribution of clustered and original jump heights.
\end{example}
\subsection{Stochastic Dominance of Subordinated Jump Sizes}
\noindent
Heuristically, one may expect that a claim 
$Z$ is, in a certain sense, larger than a claim
$X$, since
$Z$ may involve the clustering of several smaller losses from the base process. We show that this is indeed the case in the following theorem.
\begin{theorem}
    \begin{equation*}
        X\prec_{st} Z,
    \end{equation*}
where $\prec_{st}$ denotes first-order stochastic dominance.
\end{theorem}
\begin{proof}
Let $y>0$, then
    \begin{equation*}
        \mathbb{P}[Z>y]\overset{(\ref{Zdistr})}{=}\frac{\nu_Y((y,\infty))}{\psi_{\Lambda}(\lambda)}=\frac{b_\Lambda\nu_C((y,\infty))+\int_{(0,\infty)}\mu^t((y,\infty))\nu_{\Lambda}(dt)}{\psi_{\Lambda}(\lambda)}.
    \end{equation*}
    Since $C$ is a compound Poisson process, with intensity $\lambda$, 
    \begin{equation*}
        \mathbb{P}[Z>y]=\frac{b_\Lambda\lambda\mathbb{P}[X>y]+\int_{(0,\infty)}\mu^t((y,\infty))\nu_{\Lambda}(dt)}{\psi_{\Lambda}(\lambda)}
    \end{equation*}
    \begin{equation*}
    =\frac{b_\Lambda\lambda\mathbb{P}[X>y]+\int_{(0,\infty)}\mathbb{P}[\sum_{j=1}^{N_t}X_j>y]\nu_{\Lambda}(dt)}{\psi_{\Lambda}(\lambda)}.
    \end{equation*}
    Obviously for any $t\geq0$, $\{N_t\neq 0,X_1>y\}\subset\{\sum_{j=1}^{N_t}X_j>y\}$. Hence we obtain
    \begin{equation*}
        \mathbb{P}[Z>y]\geq\frac{b_\Lambda\lambda\mathbb{P}[X>y]+\int_{(0,\infty)}\mathbb{P}[N_t\neq 0]\mathbb{P}[X>y]\nu_{\Lambda}(dt)}{\psi_{\Lambda}(\lambda)}
    \end{equation*}
    \begin{equation*}
        =\frac{b_\Lambda\lambda\mathbb{P}[X>y]+\mathbb{P}[X>y]\int_{(0,\infty)}(1-e^{-\lambda t})\nu_{\Lambda}(dt)}{\psi_{\Lambda}(\lambda)}.
    \end{equation*}
    By the definition of the Laplace exponent, we know
    \begin{equation*}
        =\frac{b_\Lambda\lambda\mathbb{P}[X>y]+\mathbb{P}[X>y](\psi_{\Lambda}(\lambda)-b_{\Lambda}\lambda)}{\psi_{\Lambda}(\lambda)}=\mathbb{P}[X>y],
    \end{equation*}
    which provides the statement $X\prec_{st} Z$. 
\end{proof}
\noindent
The inequality used in the previous proof is very restrictive if the jump size $X$ admits a heavy-tailed distribution, because in that case the value $C_t$ is often triggered by a single large jump.
\\
\begin{figure}[H]
\centering
\includegraphics[width=\linewidth]{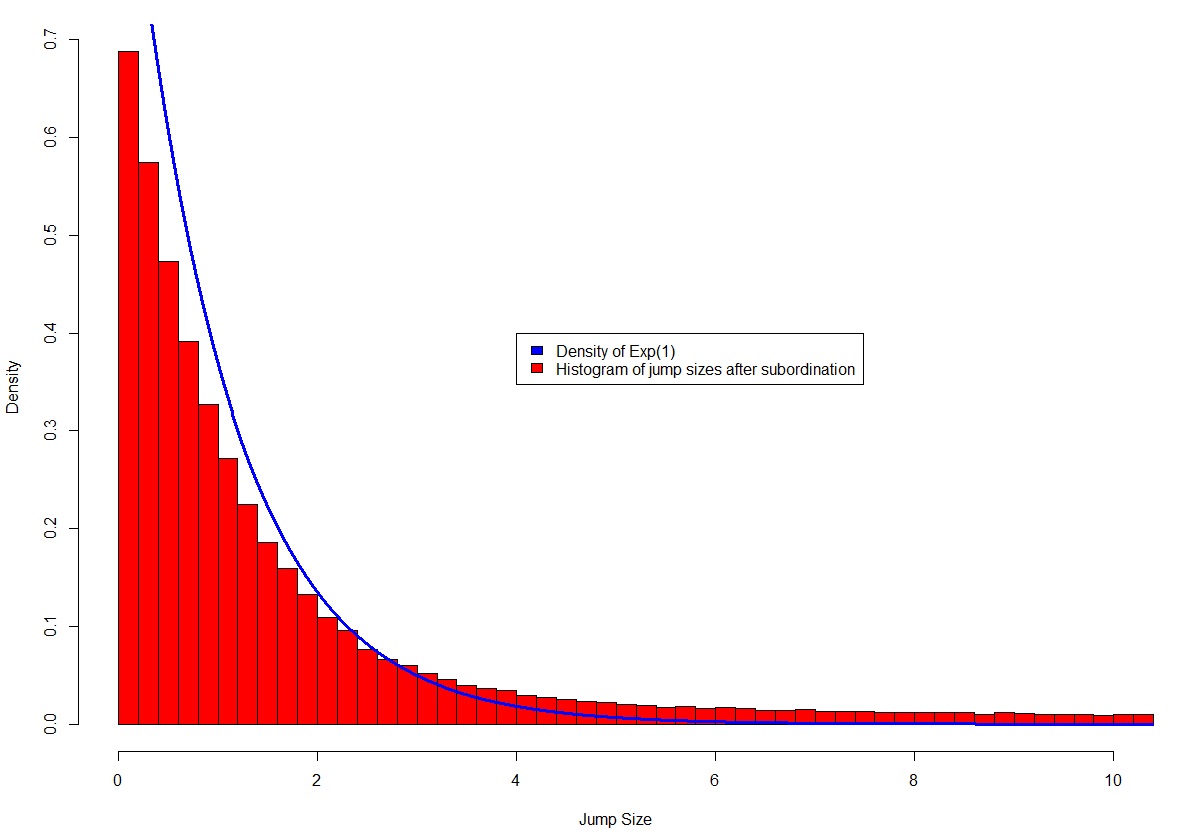}
\caption{Jump size distribution after subordination}
\label{Comparison}
\end{figure}
\noindent
Figure \ref{Comparison} illustrates the jump size distribution of a subordinated compound Poisson process.
The base process has intensity $\lambda=1$ and exponentially distributed jump sizes $X_i\sim Exp(1)$ for all $i\in\mathbb{N}$.
The Lévy subordinator is a compound Poisson process aswell, with drift $b_{\Lambda}=0.2$, intensity $\lambda_{\Lambda}=0.08$ and $Exp(0.1)$ distributed jump sizes.
Since the subordinator admits large jumps, the clustering effect of jumps in the base process is very strong.
As a result, comparing the jump size distribution of the base process and of the subordinated process, we see that the subordination causes a shift of the mass into the tail.
\\
\\
First-order stochastic dominance also allows direct conclusions to be drawn about the behaviour of risk measures. As shown in \citet{bm2006}, first-order stochastic dominance implies that
\begin{equation*}
    \rho(Z)\geq\rho(X)
\end{equation*}
for every monotone and law-invariant risk measure $\rho$.
\\
Accordingly, an increase in risk can be observed under most relevant risk measures. On the other hand, the intensity of the new process is reduced. This interplay will be analysed later.
\subsection{Heavy- and Light-Tailedness of \texorpdfstring{$Z$}{Z}}
As a next step, we want to gain a deeper understanding of the distribution of $Z$ and the effects of subordination to $(C_t)_{t\geq 0}$. Our goal is to understand the behaviour of the tail of $Z$ and to derive conditions under which the jump size $Z$ admits a light-tailed distribution. First of all, we compute the Laplace transform of $Y_t$ for all $t>0$. Since $Y$ is a Lévy process, it suffices to calculate the Laplace transform of $Y_1$, which fully characterizes the Laplace transform of the entire process $Y$.
\begin{lem}\label{lemmalaplace}
    The Laplace transform $\Psi_{Y_1}$ of $Y_1$ is given by
    \begin{equation*}
        \Psi_{Y_1}(u)=\exp(-\psi_\Lambda(\lambda[1-\Psi_X(u)])),\quad u\geq 0
    \end{equation*}
    where $\psi_{\Lambda}$ is defined as in (\ref{laplace}) and $\Psi_X$ is the Laplace transform of jump size $X$.
\end{lem}
\begin{proof}
    Denote the $\sigma$-algebra $\sigma(\Lambda_t,t\geq0)$ by $\sigma(\Lambda)$.
    For all $u\geq 0$
    \begin{equation}\label{LaplaceTrafo}
        \Psi_{Y_1}(u)=\mathbb{E}[\exp(-uC_{\Lambda_1})]=\mathbb{E}[\mathbb{E}[\exp(-u C_{\Lambda_1})|\sigma(\Lambda)]]=\mathbb{E}[\exp(\lambda\Lambda_1[\Psi_{X}(u)-1])],
    \end{equation}
    where $\Psi_X(u):=\mathbb{E}[\exp(-uX)]$. Observe that $\lambda[\Psi_X(u)-1]\leq 0$ for all $u\geq 0$, hence we write (\ref{LaplaceTrafo}) as
    \begin{equation*}
        =\exp(-\psi_{\Lambda}(\lambda[1-\Psi_{X}(u)])).
    \end{equation*}
\end{proof}
\noindent
Since $Y$ is a Lévy process, we can conclude immediately, that for any $t>0$
\begin{equation*}
    \Psi_{Y_t}(u)=\exp(-t\psi_{\Lambda}(\lambda[1-\Psi_{X}(u)])),\quad \forall u>0.
\end{equation*}
As already mentioned in the introduction of this section, we want to analyze the tail behaviour of $Y_1$ and $Z$ respectively.
A very common definition for the presence of heavy tails of a given distribution is the non-existence of the moment-generating function at a point $u>0$. In other words, the presence of heavy tails is defined by the divergence of the Laplace transform at a point $u<0$. This definition of heavy-tailedness is also used in \citet{selchscherer} and defined for example in Definition 2.2 in \citet{fosskorshunov}.
\begin{definition}
    A distribution $F$ on $\mathbb{R}$ is said to be (right-) heavy-tailed if
    \begin{equation*}
        \int_{-\infty}^{\infty}e^{ux}F(dx)=\infty,\quad \forall u>0.
    \end{equation*}
that is, if $F$ fails to possess any positive exponential moment.
\end{definition}
\noindent
Throughout the paper, we extend the Laplace exponent defined in \eqref{laplaceexpo} to negative arguments whenever this is well defined. In this case, evaluating the Laplace exponent at negative values corresponds to considering exponential moments, i.e. the moment-generating function.
We now formulate the criterion for heavy-tailedness of the distribution of $Y_1$.
\begin{lem}\label{lemmaheavytails}
    The distribution of $Y_1$ is heavy-tailed, if and only if the distribution of $\Lambda_1$ or of $X$ is heavy-tailed.
\end{lem}
\begin{proof}
    Let $u>0$. Then, if it exists, the moment-generating function $M_{Y_1}$ of $Y_1$ can be derived from the Laplace transform of $Y_1$ as
    \begin{equation}\label{moment-generating Y_1}
        M_{Y_1}(u)=\exp(-\psi_{\Lambda}(\lambda[1-M_{X}(u)])),
    \end{equation}
    where $M_X$ is the moment generating function of $X$. For $u>0$, $\lambda[1-M_{X}(u)]$ maps surjectively onto $(-\infty,0)$, therefore if the distribution of $\Lambda_1$ is heavy-tailed, i.e. $\psi_{\Lambda}(x)=-\infty$ for every $x<0$, it follows that $M_{Y_1}(u)=\infty$ for any $u>0$, which shows that $Y_1$ is heavy-tailed.
    \\
    On the other hand, let $u>0$, by (\ref{moment-generating Y_1}) it follows immediately that if $X$ admits a heavy-tailed distribution, $M_{X}(u)=\infty$ for any $u>0$ and so $M_{Y_1}(u)=\infty$, which means that $Y_1$ is heavy-tailed.
    \\
    The other direction of the theorem follows by very similar arguments. Suppose that $M_{Y_1}(u)=\infty$ for every $u>0$. This implies that $-\psi_{\Lambda}(\lambda[1-M_X(u)])=\infty$ for every $u>0$. There are two ways this could happen. First, if $\psi_{\Lambda}(u)=-\infty$ for every $u<0$, which implies that $\mathbb{E}[e^{-u\Lambda_1}]=\infty$ for every $u<0$, so $\Lambda_1$ admits a heavy-tailed distribution. Second option is that $\lambda[1-M_X(u)]=-\infty$ for any $u>0$, which immediately implies, that $M_X(u)=\infty$ for every $u>0$ and therefore $X$ admits a heavy-tailed distribution.
    \\
    This completes the proof.
\end{proof}
\newpage
\noindent
The previous proof is similar to the proof of Proposition 2.2 in \citet{selchscherer}.
Since $Y$ is a Lévy process, the statement of the preceding lemma extends to all $t\in\mathbb{R}$.
From the theorem of the previous section, we know that the subordinated process $(Y_t)_t:=(\sum_{k=1}^{N_{\Lambda_t}}X_k)_t$ has a compound Poisson representation
\begin{equation*}
    Y_t\sim \sum_{j=1}^{\tilde{N_t}}Z_j.
\end{equation*}
Hence, the previous lemma provides that $Z$ possesses a heavy-tailed distribution if and only if $X$ or $\Lambda_1$ has a heavy-tailed distribution, which can be seen through standard arguments as in the proof before.
Moreover, from the results of the preceding section we have
\begin{equation}\label{distZ}
    Z\sim \frac{\nu_Y}{\psi_{\Lambda}(\lambda)},
\end{equation}
with $\nu_Y$ as in (\ref{YLevymeasure}).
\section{Adjustment Coefficient in the Subordinated Model: Light-Tailed Jump Sizes}
In risk theory, when the jumps of the compound Poisson process have light tails, the adjustment coefficient becomes a crucial parameter for quantifying the probability of ruin.
We define the subordinated Cramér-Lundberg model via
\begin{equation}\label{CLModel}
    P_t:=u+ct-Y_t,\quad t\geq 0.
\end{equation}
As shown above, $(Y_t)_{t\geq 0}$ is again a compound Poisson process, such that the structure in the subordinated Cramér-Lundberg model and the structure in the standard Cramér-Lundberg model is kept the same. In (\ref{CLModel}), $u$ denotes the initial capital, $c$ is the premium rate and $Y_t$ models the claims side of the insurer.
\\
Throughout the whole paper, we assume that the premium rate $c$ fulfills the net profit condition, i.e.\\ $c>\mathbb{E}[Y_1]=\psi_{\Lambda}(\lambda)\mathbb{E}[Z]=\lambda\mathbb{E}[X_1]$ (see, e.g., \citet{schmidli2017}). Note that the net profit condition depends only on the expectation of $Y_t$. Consequently, it remains invariant under subordination with respect to a time-normalized subordinator. This means that the insurer's fair premium remains unaffected by the subordination.
\\
Our main interest in this article lies in the quantification of the ruin probability of the surplus process $(P_t)_{t\geq 0}$. Formally, this means quantifying the exit probability
\begin{equation*}
\Psi_P(u):=\mathbb{P}\left[\inf_{t>0}P_t<0|P_0=u\right].
\end{equation*}
\\
For the main result of this section, we require the following lemma, which shows that the Laplace exponent of a time-normalized Lévy subordinator always lies below the bisector.
\begin{lem}\label{orderinglemma}
    Let $(\Lambda_t)_t$ be a time-normalized Lévy subordinator (i.e., $\mathbb{E}[\Lambda_1]=1)$, then $\psi_{\Lambda}(u)\leq u$ for any $u\in\mathbb{R}_{+}$. Moreover, if $\psi_{\Lambda}(-s)$ exists for $s>0$, then also $\psi_{\Lambda}(-s)\leq-s$.
\end{lem}
\begin{proof}
    Let $(\Lambda_t)_t$ be a time-normalized Lévy subordinator, the Laplace exponent of $\Lambda$ is given by
    \begin{equation*}
        \psi_{\Lambda}(u)=b_{\Lambda} u+\int_0^{\infty}(1-e^{-ux})\nu_{\Lambda}(dx).
    \end{equation*}
    By the time normalization condition $\mathbb{E}[\Lambda_1]=1$, it holds that
    \begin{equation*}
        1=b_{\Lambda}+\int_{0}^{\infty}x\nu_{\Lambda}(dx).
    \end{equation*}
Therefore we get for $u\geq0$
\begin{equation*}
\psi_{\Lambda}(u)=b_{\Lambda}u+\int_{0}^{\infty}(1-e^{-ux})\nu_{\Lambda}(dx)\leq b_{\Lambda} u+\int_{0}^{\infty}ux\nu(dx)=u\left (b_{\Lambda}+\int_{0}^{\infty}x\nu_{\Lambda}(dx)\right)= u.
\end{equation*}
Suppose that $\psi_{\Lambda}(-s)$ exists for $s>0$, then we can perform the same steps, to show that $\psi_{\Lambda}(-s)\leq -s$, which proves the statement.
\end{proof}
\noindent
The previous Lemma confirms what one would expect. As we know, the subordinated process $Y$ has a compound Poisson representation with intensity $\psi_{\Lambda}(\lambda)$, where $\lambda$ is the intensity of the base process $C$. The lemma therefore implies that the intensity of $Y$ is lower, than the intensity of $C$, i.e. $\mathbb{E}[\tilde{N}_1]\leq\mathbb{E}[N_1]$, which is a consequence of the clustering in the subordinated process. Moreover the lemma will be useful for the upcoming theorem.
\newpage
\noindent
We consider the Cramér-Lundberg process $(P_t)_{t\geq0}$ as described in (\ref{CLModel}) and compare $P$ to the initial risk model
\begin{equation*}
    S_t:=u+ct-\sum_{k=1}^{N_t}X_k=u+ct-C_t.
\end{equation*}
Here, the time normalization of $\Lambda$ plays a crucial role when comparing $P$ and $S$, since the expected total claim amount remains invariant under time-normalized subordination.
\\
In the case of light-tailed jump sizes, classical Cramér-Lundberg theory applies (cf. \citet{schmidli2017}, Sections 5.5 and 5.6). We now assume that the moment-generating function of $X$ and $\Lambda_1$ exists on a sufficiently large domain. According to Lemma \ref{lemmaheavytails}, standard Cramér-Lundberg theory can therefore be applied.
Let $M_X$ and $M_Z$ be the moment-generating functions of $X$ and $Z$, respectively. We write $R$ for the non-trivial solution of
\begin{equation}\label{AdjustmentfunctionC}
    \Theta(r)=\lambda\left(M_X(r)-1\right)-cr=0
\end{equation}
and $R_\Lambda$ for the non-trivial solution of 
\begin{equation}\label{Thetalambda}
    \Theta_\Lambda(r)=\psi_{\Lambda}(\lambda)(M_Z(r)-1)-cr=0.
\end{equation}
$\Theta$ and $\Theta_{\Lambda}$ are the so-called adjustment functions of the respective Cramér-Lundberg process.
The trivial solutions are obviously given for $r=0$.
Assuming that the adjustment coefficients $R$ and $R_\Lambda$ exist and $M_{X}'(R),M_{Z}'(R_{\Lambda})<\infty$, by Theorem 5.5 in \citet{schmidli2017}, we know that the ruin probability $\Psi_S$ and $\Psi_P$ of $S$ and $P$ have the asymptotics
\begin{equation}\label{ruins}
    \Psi_S(u)\sim \frac{c-\lambda\mathbb{E}[X]}{\lambda M_X'(R)-c}e^{-Ru};\quad\Psi_{P}(u)\sim\frac{c-\lambda\mathbb{E}[X]}{\psi_{\Lambda}(\lambda)M_Z'(R_\Lambda)-c}e^{-R_\Lambda u}, \quad u\to\infty,
\end{equation}
where $u$ is the initial capital in the portfolio.
\noindent
We formulate the following.
\begin{lem}
    Suppose that the moment-generating function of $X$ exists on $[0,x)$. Moreover the moment-generating function of $\Lambda_1$ exists on $[0,l)$, then the moment-generating function of $Z$ exists on $\left[0,x\wedge(\lambda[M_X(x)-1]\wedge l)\right)$ and is given by
    \begin{equation*}
        M_Z(r)=1-\frac{\psi_{\Lambda}\left(\lambda[1-M_X(r)]\right)}{\psi_{\Lambda}(\lambda)}.
    \end{equation*}
\end{lem}
\begin{proof}
    The existence follows immediately from Lemma \ref{lemmaheavytails}.
    From Lemma \ref{lemmalaplace} we derive, that the moment-generating function of $Y_1$ is given by
    \begin{equation*}
        M_{Y_1}(r)=\exp(-\psi_{\Lambda}\left(\lambda[1-M_{X}(r)]\right)).
    \end{equation*}
    On the other hand, by the compound Poisson representation of $Y$, we know that
    \begin{equation*}
        M_{Y_1}(r)=\exp\left(\psi_{\Lambda}(\lambda)[M_{Z}(r)-1] \right).
    \end{equation*}
If the moment-generating function exists in some neighbourhood of $0$, it uniquely determines the distribution. Therefore we can combine the two equations and conclude that
\begin{equation}\label{MZ formel}
    M_Z(r)=1-\frac{\psi_{\Lambda}\left(\lambda[1-M_X(r)]\right)}{\psi_{\Lambda}(\lambda)}.
\end{equation}
The restriction to the interval $\left[0,x\wedge(\lambda[M_X(x)-1]\wedge l)\right)$ follows by \eqref{MZ formel}.
\end{proof}
\noindent
\\
Now we are able to prove the following statement for the light-tailed case. From the following result we can derive the asymptotic behaviour of the ruin probability explicitly.
\begin{theorem}
    Suppose that the adjustment coefficient $R$ of $(S_t)_{t\geq0}$ exists. Let $\Lambda$ be a time-normalized Lévy subordinator, s.t. the adjustment coefficient $R_\Lambda$ for $(P_t)_{t\geq0}$ exists. Then
    \begin{equation*}
        R_\Lambda\leq R.
    \end{equation*}
\end{theorem}
\begin{proof}
    By (\ref{AdjustmentfunctionC}) the adjustment-equation of $R$ is given by
    \begin{equation*}
        \Theta(r)=\lambda(M_X(r)-1)-cr=0.
    \end{equation*}
    By the previous Lemma and \eqref{Thetalambda}, the adjustment equation associated to $P$ is given by
    \begin{equation*}
        \Theta_\Lambda(r)=-\psi_{\Lambda}(\lambda[1-M_{X}(r)])-cr=0.
    \end{equation*}
    By Lemma \ref{orderinglemma} it follows, that
    \begin{equation*}
       \Theta_\Lambda(r)=-\psi_\Lambda(\lambda[1-M_{X}(r)])-cr\geq -\lambda(1-M_X(r))-cr=\lambda(M_X(r)-1)-cr=\Theta(r).
    \end{equation*}
It holds, that $\Theta(0)=\Theta_{\Lambda}(0)=0$. Moreover by the existence of the adjustment coefficient, $\Theta(R)=\Theta_{\Lambda}(R_{\Lambda})=0$.
By convexity of $\Theta$ and $\Theta_\Lambda$ (see for example \citet{schmidli2017}) it follows, that $R_\Lambda\leq R$ with strict inequality unless $\Lambda$ is deterministic.
\end{proof}
\noindent
The theorem shows, that in the subordinated model, the ruin probability $\Psi_P(u)$ as a function of the initial capital always decays slower than the ruin probability $\Psi_S(u)$ in the initial model. The time normalization here is crucial for ensuring comparability of these two models, as it preserves the expected value of both processes.
\begin{figure}[H]
    \centering
    \includegraphics[width=0.7\linewidth]{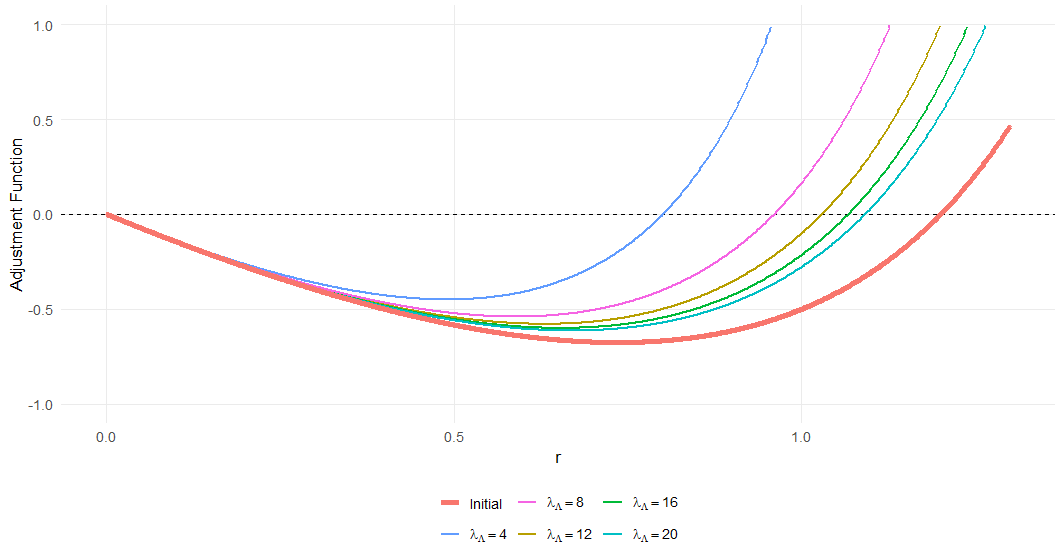}
    \caption{The adjustment function for different subordinators}
    \label{adjustmentfunction}
\end{figure}
\noindent
Figure \ref{adjustmentfunction} shows the adjustment functions for various time-normalized subordinators. As the base compound Poisson process, we have chosen the intensity $\lambda=2$ and exponentially distributed jump heights $X_i\sim Exp(2)$. The premium rate is set to $c=2.5$, which clearly satisfies the net profit condition. All considered subordinators are compound Poisson processes with different intensities $\lambda_{\Lambda}$. For simplicity, the jump sizes in the subordinator are also chosen to be exponentially distributed. Due to time normalization, the parameter of the exponentially distributed jump heights in the subordinator are therefore $\lambda_{\Lambda}$ as well. As shown in the previous theorem, the non-trivial zero points of the adjustment functions after subordination lie to the left of those of the base process. This implies that the decay of the ruin probability in the subordinated model is significantly slower than the decay of the ruin probability in the initial model (cf. (\ref{ruins})).
\newpage
\section{Asymptotic Ruin Probabilities: Heavy-Tailed Jump Sizes}
The following section is devoted to the case in which the subordinated process has a heavy-tailed jump size distribution. As established in Theorem \ref{lemmaheavytails}, this case occurs if we start with heavy-tailed jump size distributions or if the subordinator has a heavy-tailed infinitely divisible distribution. In this chapter, we focus mainly on the class of subexponential distributions. Subexponentiality offers the advantage that in certain situations the asymptotics of the ruin probability can still be quantified well. Particular attention is given to the class of distributions with regularly varying tails, for which Karamata's Theorem yields explicit asymptotics of the ruin probability. Classical Cramér-Lundberg theory is no longer applicable in this context, since the jump heights lack exponential moments and, consequently, the moment-generating function is not defined.
\\
\\
We begin by recalling some basic definitions used throughout this section.
For any cumulative distribution function $G$ with $G(0)=0$, we denote by $\overline{G}$ its right tail and for $n\in\mathbb{N}$, we write $G^{*n}$ for the $n$-fold convolution of $G$.
\subsection{Preliminary Definitions}
\begin{definition}
    A distribution $F$ with support $(0,\infty)$ is called subexponential, if for all $n\geq2$,
\begin{equation*}
    \lim_{x\to\infty}\frac{\overline{F^{*n}}(x)}{\overline{F}(x)}=n.
\end{equation*}
We denote the class of subexponential distributions by $\mathcal{S}$.
\end{definition}
\noindent
Another very important definition for this chapter is the following.
\begin{definition}
    \begin{enumerate}
  \item A positive, Lebesgue measurable function \(L\) on \((0,\infty)\) is slowly varying (at \(\infty\)) if
  \[
    \lim_{x\to\infty} \frac{L(t\,x)}{L(x)} \;=\; 1,
    \quad t>0.
  \]
  \\
  \item A positive, Lebesgue measurable function \(h\) on \((0,\infty)\) is regularly varying (at \(\infty\)) of index \(\alpha\in\mathbb{R}\) if
  \[
    \lim_{x\to\infty} \frac{h(t\,x)}{h(x)} \;=\; t^{\alpha},
    \quad t>0.
  \]
\end{enumerate}
\end{definition}
\noindent
Clearly, the case of a slowly varying function corresponds to a regularly varying function with the index $\alpha=0$.  
In the following, the cumulative distribution function of $\Lambda_1$ is denoted by $F_{\Lambda_1}$ and analogously the CDF of $X$ is denoted by $F_{X}$, where $X$ and $\Lambda_1$ are defined as before. We also write $\overline{F}_{\Lambda_1}$ and $\overline{F}_{X}$ for the corresponding right tails, i.e.
\begin{equation*}
    \overline{F}_{\Lambda_1}(x):=1-F_{\Lambda_1}(x),\quad \overline{F}_{X}(x):=1-F_{X}(x).
\end{equation*}
Our interest lies in the class of distribution functions whose tail $\overline{F_{\cdot}}$ is a regularly varying function with some index $-\alpha<0$. As shown in \citet{ekm2013}, distributions with this property belong to the class $\mathcal{S}$.
\noindent
In the following, we write $f(x)\sim g(x)$ for functions $f$ and $g$, if
\begin{equation*}
    \lim_{x\to\infty} \frac{f(x)}{g(x)}=1.
\end{equation*}
\subsection{Subexponential Initial Jump Size}
First, we consider the case where the subordinator has a light-tailed distribution and $X$ is subexponential.
\begin{theorem}
    Suppose that $\Lambda_1$ is time-normalized and light-tailed. Moreover $\mathcal{L}(X|\mathbb{P})\in\mathcal{S}$. Then $Z\in\mathcal{S}$ and \begin{equation*}
        \mathbb{P}[Z>x]\sim\frac{\lambda}{\psi_{\Lambda}(\lambda)}\mathbb{P}[X>x].
    \end{equation*}
\end{theorem}
\begin{proof}
Let $t>0$, $\varepsilon>0$. We calculate the expectation
    \begin{equation*} 
        \mathbb{E}[e^{\varepsilon N_{\Lambda_t}}]=\mathbb{E}[\mathbb{E}[e^{\varepsilon N_{\Lambda_t}}|\sigma(\Lambda_t)]].
    \end{equation*}
Since $N$ is a standard Poisson process we get
\begin{equation}\label{momenterzeugende}
    =\mathbb{E}[e^{\lambda\Lambda_t(e^{\varepsilon}-1)}].
\end{equation}
Since $\Lambda_t$ is light-tailed for any $t\geq 0$, we can choose
$\varepsilon>0$ sufficiently small so that the expression in
\eqref{momenterzeugende} is finite. Hence $N_{\Lambda_t}$ is light-tailed for any $t\geq0$.
\\
By Theorem 1 in \citet{schmidli1999} it follows that for any $t>0$, the distribution of $Y_t:=\sum_{i=1}^{N_{\Lambda_t}}X_i$ belongs to $\mathcal{S}$.
\\
Moreover by \citet{schmidli1999}, for $x>0,$
\begin{equation}\label{eq1}
    \mathbb{P}[Y_t>x]=\mathbb{P}[\sum_{i=1}^{N_{\Lambda_t}}X_i>x]\sim \lambda t\mathbb{P}[X>x].
\end{equation}
By \eqref{equaldistr} and Theorem A.3.19. in \citet{ekm2013} it also follows that
\begin{equation}\label{eq2}
    \frac{\mathbb{P}[Y_t>x]}{\mathbb{P}[Z>x]}\sim\psi_{\Lambda}(\lambda)t.
\end{equation}
Combining \eqref{eq1} and \eqref{eq2}, we obtain
\begin{equation}\label{TailasymptotikZX}
    \mathbb{P}[Z>x]\sim\frac{\lambda}{\psi_{\Lambda}(\lambda)}\mathbb{P}[X>x].
\end{equation}
By $X\in\mathcal{S}$, \eqref{TailasymptotikZX}, and Lemma A3.15 in \citet{ekm2013} it follows that $Z\in\mathcal{S}$,
which proves the statement.
\end{proof}
\noindent
Equation \eqref{TailasymptotikZX} admits a very natural interpretation: on average, $\frac{\lambda}{\psi_{\Lambda}(\lambda)}$ claims of the original compound Poisson process with jump size $X$ are aggregated into a single jump in the subordinated compound Poisson process, whose jump size is $Z$. 
\subsection{Regularly Varying Subordinator}
In this subsection we use a Lévy subordinator $\Lambda$, such that $\overline{F_{\Lambda_1}}$ is regularly varying. Due to time normalization, it is necessary that the index $-\rho$ of the regularly varying tail is smaller than $-1$ to ensure that the expected value of $\Lambda_1$ exists. We formulate the following.
\begin{theorem}
    Let $\Lambda$ be a time-normalized Lévy subordinator, s.t. $\overline{F_{\Lambda_1}}$ is regularly varying with index $-\rho<-1$, i.e. $\overline{F_{\Lambda_1}}(x)\sim x^{-\rho}L(x)$, for some slowly varying $L$. Furthermore let $\overline{F_X}(x)\in o(x^{-\rho}L(x))$, then $\overline{F_Z}$ is regularly varying with index $-\rho$.
\end{theorem}
\begin{proof}
    It holds that for $x\in\mathbb{N}$
    \begin{equation*}
        \mathbb{P}[N_{\Lambda_1}=x]=\int_{0}^{\infty}\mathbb{P}[N_r=x]dF_{\Lambda_1}(r)=\int_{0}^{\infty}\frac{(\lambda r)^x}{x!}e^{-\lambda r}dF_{\Lambda_1}(r).
    \end{equation*}
    Substituting $z:=\lambda r$, we get
    \begin{equation*}
        =\int_{0}^{\infty}\frac{z^{x}}{x!}e^{-z}dF_{\Lambda_1}(\frac{z}{\lambda})
    \end{equation*}
    for any $x\in\mathbb{N}$. Therefore $N_{\Lambda_1}$ is mixed Poisson distributed. Defining $\mathcal{L}(z):=\lambda^{\rho}L(z)$, by slow variation of $L$
    \begin{equation*}
        \overline{F_{\Lambda_1}}(\frac{z}{\lambda})\sim z^{-\rho}\lambda^{\rho}L(\frac{z}{\lambda})\sim z^{-\rho}\mathcal{L}(z),\quad z\to\infty.
    \end{equation*}
    By Proposition 8.4 and Corollary 8.5 in \citet{grandell1997}, it follows that
    \begin{equation*}
        \mathbb{P}[N_{\Lambda_1}>z]\sim\overline{F_{\Lambda_1}}(\frac{z}{\lambda})\sim z^{-\rho}\mathcal{L}(z),\quad z\to\infty.
    \end{equation*}
    Since $\overline{F_X}(x)\in o(x^{-\rho}L(x))$ it holds per definition that
    \begin{equation*}
        \lim_{x\to\infty}\overline{F_X}(x)x^{\rho}\frac{1}{L(x)}=0.
    \end{equation*}
    Furthermore we know
    \begin{equation*}
        \lim_{x\to\infty}\mathbb{P}[N_{\Lambda_1}>x]x^{\rho}\frac{1}{L(x)}=\lim_{x\to\infty}x^{-\rho}\mathcal{L}(x)x^{\rho}\frac{1}{L(x)}=\lambda^{\rho}.
    \end{equation*}
    By Theorem 1.3. in \citet{stam1973},
    \begin{equation*}
        \lim_{z\to\infty}\mathbb{P}[Y_1>z]z^{\rho}\frac{1}{L(z)}=(\lambda\mathbb{E}[X])^{\rho},
    \end{equation*}
    which shows that
    \begin{equation*}
        \mathbb{P}[Y_1>z]\sim (\lambda\mathbb{E}[X])^{\rho}z^{-\rho}L(z).
    \end{equation*}
    We conclude by Theorem A3.19. in \citet{ekm2013}, that
    \begin{equation*}
        \overline{F_Z}(x)\sim \frac{1}{\psi_{\Lambda}(\lambda)}(\lambda\mathbb{E}[X])^\rho z^{-\rho}L(z).
    \end{equation*}
\end{proof}
\noindent
In the context of the Cramér-Lundberg model and ruin probabilities, the case of regularly varying tails of the claim size distribution is very comfortable. By applying Karamata´s Theorem, the asymptotics of the ruin probability can be obtained explicitly.
\begin{corollary}
    Let $\Lambda$ be a time-normalized Lévy subordinator, s.t. $\overline{F_{\Lambda_1}}(x)\sim x^{-\rho}L(x)$ for a slowly varying $L$ and $-\rho<-1$. Furthermore $\overline{F_X}(x)\in o(x^{-\rho}L(x))$. Consider the subordinated Cramér-Lundberg process
    \begin{equation*}
        P_t=u+ct-\sum_{k=1}^{N_{\Lambda_t}}X_k,
    \end{equation*}
    then the ruin probability of $P_t$ is asymptotically given by
    \begin{equation*}
        \Psi_P(u)\sim\frac{1}{c-\lambda\mathbb{E}[X]}\frac{1}{\rho-1}(\lambda\mathbb{E}[X])^{\rho} u^{-\rho+1}L(u),\quad u\to\infty.
    \end{equation*}
\end{corollary}
\begin{proof}
    Since $\overline{F_{\Lambda_1}}$ is regularly varying with index $-\rho<-1$, we know from the previous theorem, that
    \begin{equation*}
    \overline{F_Z}(z)\sim \frac{1}{\psi_{\Lambda}(\lambda)}(\lambda\mathbb{E}[X])^{\rho}z^{-\rho}L(z)
    \end{equation*}
    for $z\to\infty$.
    \\
    The integrated tail of $\overline{F_Z}$ is defined by
    \begin{equation*}
        \overline{F_{I}^{Z}}(x):=\frac{1}{\mathbb{E}[Z]}\int_{x}^{\infty}\overline{F_Z}(y)dy.
    \end{equation*}
    By applying the Karamata's Theorem, we obtain due to the regular variation of $\overline{F_{Z}}$,
    \begin{equation*}
        \frac{1}{\mathbb{E}[Z]}\int_x^{\infty}\overline{F_{Z}}(y)dy\sim\frac{1}{\rho-1}\frac{1}{\psi_{\Lambda}(\lambda)}(\lambda\mathbb{E}[X])^{\rho}\frac{1}{\mathbb{E}[Z]}x^{-\rho+1}L(x)=\frac{1}{\rho-1}{(\lambda\mathbb{E}[X])^{\rho-1}}x^{-\rho+1}L(x).
    \end{equation*}
    Applying Theorem 1.3.6 in \citet{ekm2013}, we obtain
    \begin{align}
        \psi_P(u)&\sim \frac{\lambda\mathbb{E}[X]}{c-\lambda\mathbb{E}[X]}\frac{1}{\rho-1}(\lambda\mathbb{E}[X])^{\rho-1} u^{-\rho+1}L(u)\\&=\frac{1}{c-\lambda\mathbb{E}[X]}\frac{1}{\rho-1}(\lambda\mathbb{E}[X])^{\rho}u^{-\rho+1}L(u),\quad u\to\infty.
    \end{align}
\end{proof}
\noindent
This result particularly striking. According to the preceding theorem, regular variation is inherited from the subordinator to the jump size $Z$ of the subordinated process. Consequently, the asymptotic behaviour of the insurer's ruin probability can be quantified using Karamata's Theorem. This represents a risk that an insurer cannot afford to ignore. 
\\
\\
As noted earlier, time-normalized subordination is solely a random change in the timing of loss occurrences and does not alter the expected value. Nevertheless, the probability of ruin can be seriously affected by the clustering of losses (as in Nat-Cat events), even if the individual losses appear very moderate. The following example illustrates this.
\noindent
\begin{example}
Let $C_t:=\sum_{i=1}^{N_t}X_i$ be a compound Poisson process with $X_1\sim Exp(1)$ and $\mathbb{E}[N_1]=1$. For the subordinator $\Lambda$ we define
\begin{equation*}
    \Lambda_t:=0.9t+\sum_{i=1}^{\bar{N_t}}K_i,
\end{equation*}
where $\bar{N}$ is a Poisson process with intensity $1$ and $(K_j)_{j\in\mathbb{N}}$ are i.i.d. positive random variables with $K_1\sim Par(\frac{0.1\varepsilon}{1+\varepsilon},1+\varepsilon)$ for $\varepsilon>0$. Here, the first parameter denotes the scale parameter of the Pareto distribution, while the second parameter represents the shape parameter. The parameters are chosen in such a way, that $\Lambda$ fulfills the time normalization assumption, i.e. $\mathbb{E}[\Lambda_1]=1$.
\\
We compare the ruin probabilities of
\begin{equation*}
    S_t:=u+ct-\sum_{i=1}^{N_t}X_i
\end{equation*}
and
\begin{equation*}
    P_t:=u+ct-\sum_{i=1}^{N_{\Lambda_t}}X_i.
\end{equation*}
\\
We know by Propostition 1.5 in section X.1 and by Lemma 2.2 in section X.2 in \citet{asmussenalbrecher}, that
\begin{equation*}
    \mathbb{P}[\Lambda_1>x]=\mathbb{P}[0.9+\sum_{i=1}^{\bar{N_t}}K_i>x]\sim \mathbb {P}[\sum_{i=1}^{\bar{N_t}}K_i>x]\sim \mathbb{P}[K_1>x]=\left (\frac{0.1\varepsilon}{(1+\epsilon)x}\right )^{1+\varepsilon},
\end{equation*}
for $x\to\infty$.
Therefore $\Lambda_1$ is regularly varying with index $-(1+\varepsilon)$. By the previous corollary, we have
\begin{equation*}
\Psi_{P}(u)\sim\frac{1}{c-1}\frac{1}{\varepsilon}u^{-\varepsilon}\left(\frac{0.1\varepsilon}{1+\varepsilon}\right)^{1+\varepsilon}.
\end{equation*}
Moreover for $S$, standard Cramér-Lundberg theory yields
\begin{equation*}
    \Psi_{S}(u)\sim\frac{1}{c}e^{-(1-\frac{1}{c})u}.
\end{equation*}
Note that $\mathbb{E}[S_t]=\mathbb{E}[P_t]$ for all $t\geq0$. Nevertheless, asymptotic behaviour of the ruin probabilities differs dramatically. For $\varepsilon>0$ small, $\Psi_{P}(u)$ exhibits polynomial decay that can be made arbitrarily slow, whereas $\Psi_{S}(u)$ always decays exponentially fast.
\end{example}
\newpage
\section{Concluding Remarks}
In this paper, we have developed a model capable of incorporating cumulative losses, such as natural catastrophe events, into an insurance portfolio consisting of individual claims. The insurer's initial loss process was modeled by a compound Poisson process. Nat-Cat losses were then introduced by means of a stochastic time change, which was given by a Lévy subordinator.
The Lévy subordinator was chosen so that it runs as fast as calendar time on average. This time normalization enabled a meaningful quantitative comparison between the subordinated and the initial portfolio, since under this condition the subordination does not alter the expectation of the risk process. Moreover, time normalization ensured that subordination represents purely a random distortion of the claim occurrence times.
Nevertheless, we have shown, that this change in the timing has a substantial impact on the asymptotic behaviour of the ruin probability.
\\
\\
We proved that the insurer's time-changed claims process remains a compound Poisson process and derived conditions under which the subordinated process exhibits light-tailed or heavy-tailed jump distributions. We were able to quantify the risk of ruin in both cases. Moreover we have shown that subordination is always bad for the insurer in the case of light tails. In the case of heavy tails our analysis mainly focused on the jump distributions with regularly varying tails. We established conditions under which the subordinated process inherits regular variation and using Karamata's theorem, obtained asymptotic results for the ruin probability in this case. We have shown that also in the case of regular variation, under very mild conditions, the effect of subordination on the ruin probability is negative.
\\
\\
As shown in the last section, if we start with an initial model whose jump distribution is light-tailed and apply a time-normalized Lévy subordinator with regularly varying tails, the ruin probability decays only polynomially. Depending on the choice of subordinator, this polynomial decay can be made arbitrarily slow. It is important to note that, without subordination, the initial model always exhibits an exponential decay of the ruin probability. The dramatic difference in the asymptotic behaviour arises from the fact that, in the subordinated model, losses occur in clusters rather than being evenly spread over the insurance horizon, as in the initial model.
\newpage
\printbibliography
\end{document}